\documentclass[twoside,leqno,twocolumn]{article}
\usepackage[letterpaper]{geometry}

\usepackage{ltexpprt}

\usepackage{pgfplots}
\usepackage{pgfplotstable}
\pgfplotsset{compat=1.8}
\usepackage{subcaption}
\usepackage{xspace} 
\usepackage{xcolor}
\usepackage{amsmath}
\usepackage{amssymb}
\usepackage{algorithm}
\usepackage{algorithmicx}
\usepackage[normalem]{ulem} 

\usepackage{cite}
\usepackage{float}
\usepackage{siunitx}
\usepackage{graphicx}
\usepackage{thm-restate}
\usepackage{blkarray}
\usepackage{balance}  
\usepackage{booktabs} 

\usepackage{marginnote} 
\usepackage{url}
\usepackage{enumitem}
\graphicspath{{./fig/}}
\usepackage{mathtools}
\usepackage[noend]{algpseudocode}

\usepackage{comment}
\usepackage{color}
\usepackage{cleveref}
\usepackage{tikz}
\usetikzlibrary{fadings}

\usepackage{dirtytalk}

\newcommand{\punt}[1]{}

\newtheorem{definition}{Definition}

\newtheorem{observation}{Observation}
\newtheorem{question}{Question}

\newcommand{\defn}[1]       {{\textit{\textbf{\boldmath #1}}}}
 
\newcommand{\poly}{\mbox{poly}}
\newcommand{\polylog}{\mathrm{polylog}}
\newcommand{\sort}{\mathrm{sort}}
\newcommand{\scan}{\mathrm{scan}}
\newcommand{\permute}{\mathrm{permute}}
\newcommand{\loglog}{\mathrm{loglog}}
\newcommand{\vsketch}{\mathrm{vsketch}}
\newcommand{\ksketch}{\mathrm{ksketch}}

\newcommand{\graph}{\mathcal{G}}
\newcommand{\nodes}{\mathcal{V}}
\newcommand{\edges}{\mathcal{E}}
\newcommand{\nodesize}{V}
\newcommand{\edgesize}{E}
\newcommand{\graphstream}{\xi}
\newcommand{\streamlength}{N}
\newcommand{\sketch}{\mathcal{S}}
\newcommand{\kconsketch}{\mathcal{K}}
\newcommand{\blocksize}{B}
\newcommand{\memsize}{M}
\newcommand{\disksize}{D}
\newcommand{\streamelement}{s}

\newcommand{\cert}{\mathcal{H}}

\newcommand{\nodegroup}{\mathcal{U}}

\newcommand{\boruvka}{Bor\r{u}vka\xspace}

\newcommand{\algname}{\textsc{StreamingCC}\xspace}
\newcommand{\externalcc}{\textsc{ExtSketchCC}\xspace}
\newcommand{\modelname}{external semi-streaming\xspace}
\newcommand{\Modelname}{External Semi-Streaming\xspace}

\newcommand{\softO}{\widetilde O}

\renewcommand{\subparagraph}[1]{\smallskip
\noindent
\emph{#1 }}

\renewcommand{\paragraph}[1]{\vspace{0.5em} \noindent \textbf{#1}}

\newcommand{\argmax}{\text{argmax}\xspace}

\newcommand{\green}[1]{\colorbox{green}{#1}}
\newcommand{\yellow}[1]{\colorbox{yellow}{#1}}
\newcommand{\red}[1]{\colorbox{red}{#1}}

\newcommand{\etal}{\text{et al}.\xspace}

\renewcommand{\epsilon}{\varepsilon}

\renewcommand{\eqref}[1]          {Eq.~\ref{eq:#1}}

\begin{document}

\title{The Case for External Graph Sketching}

\author{Michael A.~Bender\\ \small{Stony Brook University and RelationalAI} \and Mart\'{\i}n Farach-Colton\\ \small{New York University} \and Riko Jacob\\ \small{IT University of Copenhagen} \and Hanna Koml\'os\\ \small{New York University} \and David Tench \\ \small{Lawrence Berkeley National Laboratory} \and Evan T. West\\\small{Stony Brook University}}

 \date{}

\maketitle
Algorithms in the data stream model use $O(polylog(N))$ space to compute some property of an input of size $N$, and many of these algorithms are implemented and used in practice. However, sketching algorithms in the graph semi-streaming model use $O(V polylog(V))$ space for a $V$-vertex graph, and the fact that implementations of these algorithms are not used in the academic literature or in industrial applications may be because this space requirement is too large for RAM on today's hardware.

In this paper we introduce the external semi-streaming model, which addresses the aspects of the semi-streaming model that limit its practical impact. In this model, the input is in the form of a stream and $O(V polylog(V))$ space is available, but most of that space is accessible only via block I/O operations as in the external memory model. The goal in the external semi-streaming model is to simultaneously achieve small space and low I/O cost. 

We present a general transformation from any vertex-based sketch algorithm to one which has a low sketching cost in the new model. We prove that this automatic transformation is tight or nearly (up to a $O(\log(V))$ factor) tight via an I/O lower bound for the task of sketching the input stream.

Using this transformation and other techniques, we present external semi-streaming algorithms for connectivity, bipartiteness testing, $(1+\epsilon)$-approximating MST weight, testing k-edge connectivity, $(1+\epsilon)$-approximating the minimum cut of a graph, computing $\epsilon$-cut sparsifiers, and approximating the density of the densest subgraph. These algorithms all use $O(V poly(\log(V), \epsilon^{-1},k)$ space. For many of these problems, our external semi-streaming algorithms outperform the state of the art algorithms in both the sketching and external-memory models. 

\section{Introduction}
The streaming model has been widely successful in both the theory and systems literature for a variety of reasons. 
In this model, the input is presented as a arbitrarily-ordered stream of updates, and the challenge is to design algorithms that compute properties of the input in small (ideally polylogarithmic) space.
On the theory side, it is an elegant and simplified model that allows for the development of interesting upper bounds without getting bogged down in the details of the hardware.
Despite the relative simplicity of the model, it has also been successful in the systems literature because it does capture important aspects of real computers. Specifically, it captures the idea that caches are small and fast and streams of data arrive quickly and are too big to store. 
The streaming model succeeded because it hit a sweet spot between the elegance needed for theoretical results and capturing the right hardware constraints for designing software. 

However, not all streaming problems can be solved in the streaming model. For example, most graph-theoretic problems have outputs that are by themselves too large to store in the polylogarithmically sized RAM specified by the streaming model. The graph semi-streaming model~\cite{semistreaming1,semistreaming2} was introduced in order to bridge this gap. Specifically, in the semi-streaming model we assume that we have $O(\nodesize \polylog\nodesize)$ space, where $\nodesize$ denotes the number of vertices in the input graph. 
The input stream is a sequence of edge insertions (and possibly deletions).

The semi-streaming model has proven to be fertile soil for theoretical results for both upper and lower bounds. For example, there is a rich literature for addressing a long list of graph problems~\cite{Ahn2012,Ahn2012_2,nelson2019optimal,semistreaming1,semistreaming2,KapralovLMMS13,GuhaMT15,AhnGM13,ChitnisCEHMMV16,AssadiKL16,McGregorVV16,KapralovW14,McGregorTVV15,CrouchMS13,McGregorV15,clustering,pagh2012colorful,chitnis2014parameterized,mightaswell,assadi2023tight}, as well as computational geometric problems~\cite{geom_I04, geom_BGMS18, geom_CJKV22, geom_WY22, geom_BCEG04, geom_FS05}, and hypergraph problems~\cite{GuhaMT15,mcgregor2021maximum,emek2016semi,mcgregor2019better,assadi2016tight,bateni2017almost,streamsetbounds}. The semi-streaming model is elegant and captures something exciting about the structure of graph problems and their algorithms.

In the general case where the stream contains deletions, all known space-efficient algorithms are \defn{linear sketches}. These algorithms generate a random linear projection of the input that can be stored in $O(\nodesize \polylog\nodesize)$ space. Moreover, there is strong theoretical evidence that linear sketches are universal; i.e., that any space-efficient single-pass semi-streaming algorithm with constant success probability can be formulated as a linear sketch~\cite{mightaswell}.

There is a corresponding need in industry and large-scale science to process massive, dynamic graphs. A recent survey by Sahu \etal~\cite{ubiquitous} of industrial uses of graph algorithms indicates that a majority of industry respondents need to process massive (at minimum multi-billion-edge) graphs, and a majority work with graphs that change over time. Scientific applications include metagenome assembly~\cite{metagenomics1,metagenomics2}, large-scale clustering~\cite{clustering5,clustering6,clustering1}, and tracking social network communities that change as users add or delete friends~\cite{lee2014social,dynamic_social}.

Despite its great theoretical success and a wealth of potential applications, the semi-streaming model has not yielded a corresponding applied literature. The reason is straightforward -- for practical purposes $O(\nodesize\polylog\nodesize)$ space is enormously larger than RAM and will remain so for the foreseeable future. 
In Appendix~\ref{app:case} we illustrate this problem with a case study on the k-connectivity sketch of Ahn \etal~\cite{Ahn2012}, one of the simplest and smallest graph sketch algorithms. We show that the logarithmic and constant factors in the space complexity of this algorithm are large enough that it will not save space until we have RAM sizes in the hundreds of TBs. 
We note in the appendix that due to a lower bound of Nelson and Yu~\cite{nelson2019optimal}, this asymptotic space complexity cannot be significantly improved.
Additionally, the k-connectivity sketch is a subroutine for many other semi-streaming algorithms such as minimum cut, spectral sparsification, and minimum spanning tree~\cite{Ahn2012}. This suggests that the problem illustrated in our case study is not an isolated one 
but rather a general limitation of many semi-streaming algorithms. 

In this paper, we show that not all hope is lost. Our optimism is based on a reexamination of hardware trends. We notice that the bandwidth of high-speed storage systems is now so high that the cost of random access to RAM is comparable to the cost of sequential access to storage. On the other hand, such high-speed storage is expensive and limited in size. So one of the contributions of this paper is a modification of the semi-streaming model based on modern hardware. For a more detailed description of hardware that we use to reach these conclusions, see Appendix~\ref{app:disk}.

This observation about bandwidths has an interesting algorithmic implication. Many if not most advanced semi-streaming algorithms use hashing to keep sketches in RAM, and therefore the random-access bandwidth of RAM is an upper bound on their performance.\footnote{Sequential-access RAM is significantly faster (an order of magnitude) but existing sketch data structures do not perform updates sequentially.}
So a hypothetical semi-streaming algorithm that made only sequential accesses could be run on storage at the same speed that traditional semi-streaming algorithms can be run on RAM.

\begin{question}
    Is there a way to redesign a semi-streaming algorithm that uses random RAM accesses to perform the same computation using sequential accesses instead? 
\end{question} 

This notion of algorithms limited to sequential accesses on disk is captured by the well-studied external-memory model~\cite{Vitter01}. It assumes RAM of size $M$ and disk of unbounded size. Words in RAM can be accessed for free, but disk is accessed in blocks of size $\blocksize = o(\memsize)$ and each block access costs one dick access (I/O). The goal is to minimize I/O cost.

So our hardware observations seem to reflect some aspects of both of these models: because semi-streaming algorithms are too large for modern RAM, we hope instead to run them on modern high-speed storage devices. The block-access constraint of the external-memory model captures the need to design algorithms that make sequential accesses for good performance. The space constraint of semi-streaming captures the fact that modern high-speed storage devices are large enough to store semi-streaming data structures but not the entire input graph.

In this paper, we formalize this combination of semi-streaming and external memory to provide a theoretical model that allows for interesting algorithmic development while also holding out the hope that more of the  good ideas that have already been developed in the semi-streaming literature can find their way to practical relevance. 

\paragraph{The \modelname model.} 
As in the semi-streaming model, graph updates in the form of edge insertions or deletions are received in a stream, and the total space available is $O(\nodesize \polylog(\nodesize))$. However, there is an additional constraint on the \emph{type} of memory available for computation: only $\memsize = \Omega(\polylog(\nodesize))$ and $\memsize=o(\nodesize)$ RAM is available, as in the data stream model, and $\disksize = O(\nodesize \polylog(\nodesize))$ and $\disksize = o(\nodesize^2)$ disk space is available. 
As in the external-memory model, a word in RAM is accessed at no cost, and disk is accessed in blocks of $\blocksize = o(\memsize)$ words at a cost of a single I/O. 
The algorithmic challenge in the new model is to minimize the I/O  complexity (of ingesting stream updates and computing solutions to queries) in addition to satisfying the typical limited-space requirement of the data stream model. 

\paragraph{The technical challenge.}
Of course, any existing semi-streaming algorithm can be run in the new model, in the sense that the data structures can be stored on disk. However, since most existing graph-sketching algorithms use techniques such as hashing to random locations \cite{mcgregor2014graph,Ahn2012,KapralovLMMS13,GuhaMT15,assadi2023tight}, most accesses that the algorithm makes will be random accesses. So this approach falls short because the I/O cost is too high.

Let us get more specific about what a good algorithm in the new model would look like. Any graph sketch algorithm first compresses a large graph stream into a small sketch (we call this \emph{sketching the input stream}), and then extracts an answer to the problem from the small sketch. We want an algorithm that completes both of these steps at low I/O cost and uses low space throughout. Specifically an ideal algorithm would have the following properties:

\begin{enumerate}
    \item \label{port:ingest} \textbf{Sketching cost.} The cost of sketching the input stream is at most the cost of sorting it. We choose sorting as our target complexity because it is a natural lower bound for most non-trivial external-memory problems.

    \item \label{port:query} \textbf{Extraction cost.} The cost of computing the desired property of the input graph from the sketch is no more than the cost of computing the property on a sparsified graph via the best existing external-memory algorithm. 
    \item \label{port:space} \textbf{Space.} The space required is the same as the best existing graph sketch for the problem.
\end{enumerate}

\subsection{Overview of Results}
Our first result is a definition of the \modelname model.
We also present \modelname graph sketch algorithms  for classical problems: connectivity, hypergraph connectivity, minimum cut, cut sparsification, bipartiteness testing, minimum spanning tree, and densest subgraph. These algorithms meet or nearly meet the above list of properties that ideal \modelname algorithms should have; see Section~\ref{sec:results} for a detailed discussion.

Moreover, we show how to transform a graph sketch algorithm that was not designed with \modelname in mind into one that sketches the input stream with an I/O cost roughly equivalent to that of permuting the stream. This transformation does not increase the space cost. This transformation applies to a large~\cite{mcgregor2014graph,Ahn2012,KapralovLMMS13,GuhaMT15,assadi2023tight} class of sketches which are called \defn{vertex-based sketches} (see Section~\ref{sec:preliminaries}).

We complement these upper bounds with I/O lower bounds for sketching input streams in external memory via a reduction from sparse matrix-dense vector multiplication. For several of the problems we consider, the upper and lower bounds match.

The \modelname algorithms we present in this paper match the space costs of the best existing semi-streaming algorithms for these problems, but have much lower I/O complexity. Interestingly, several of our \modelname algorithms have I/O costs comparable to or better than existing external-memory algorithms for the same problems. There are two important consequences of these results. 

\paragraph{Graph-sketching algorithms via external-memory techniques.}
Practical graph-sketching algorithms will have to use disk, but achieving I/O efficiency is not overly painful. 
By leveraging external-memory techniques, we design graph-sketching algorithms with low I/O cost. In fact, the I/O cost of most of these algorithms is competitive with the best existing external-memory algorithms even though our algorithms use limited space and only have stream access to the input. Even better, these results show that the algorithm designer who desires practical semi-streaming algorithms need not throw out existing techniques from the semi-streaming literature. They simply require additional work to make them efficient in the new model. 

\paragraph{External-memory graph algorithms via graph sketching.} Graph sketching is a fruitful technique for designing external-memory graph algorithms because one can exploit the data locality of sketches to minimize I/Os. For example, we present the first nontrivial external-memory algorithms for hypergraph connectivity, cut sparsification, and densest subgraph. 
Our algorithms for $k-$connectivity and $\varepsilon-$approximate min cut also outperform the best existing algorithms for these problems for some parameter settings (see Section~\ref{sec:results} for details).

\section{Preliminaries}
\label{sec:preliminaries}

\paragraph{Graphs and hypergraphs.}
For graph $\graph = (\nodes, \edges)$ let $\nodesize = |\nodes|$ and $\edgesize = |\edges|$. In this paper we consider only undirected graphs. For convenience, we number the nodes in the graph arbitrarily from $0$ to $\nodesize-1$ and adopt the convention that the node ID of $u$ is less than node ID of $v$ for edge $e = (u,v)$. We refer to $u$ as the \defn{left endpoint} of $e$ and $v$ as the \defn{right endpoint} of e. We sometimes consider weighted graphs where each edge has a weight $w(e) \geq 0$.
We define the following notation for graph properties: let $\lambda(\graph)$ denote the minimum weight cut (or equivalently the \defn{edge connectivity}) of graph $\graph$, and then $|\lambda(\graph)|$ denote the weight of that cut.  Similarly, let $\lambda_{st}(\graph)$ denote the minimum weight $s-t$ cut in $\graph$, let $\lambda_e(\graph)=\lambda_{uv}$ denote the $u-v$ cut for edge $e = (u,v)$, and let $|\lambda_{st}(\graph)|$, $|\lambda_e(\graph)|$ denote their respective weights. For any $S \subset \nodes$ let $\lambda_{S}(\graph)$ denote the $(S, \nodes \setminus S)$ cut in (hyper)graph $\graph$.

A hypergraph $\graph$ is specified by a set of vertices $\nodes$ and a set $\edges$ of subsets of $\nodes$ called hyperedges. We assume all hyperedges have cardinality at most $r$ for some fixed $r$. Hypergraphs are a generalization of graphs, which correspond to the special case where each hyperedge has cardinality two. Let a spanning graph $T = (\nodes, \edges_T)$ of a hypergraph be a subgraph such that $|\lambda_S(T)| \geq \min\{1, |\lambda_{S}(\graph)|\}$ for every $S \subset \nodes$.

\paragraph{Semi-streaming model.}
In the \defn{graph semi-streaming} model~\cite{semistreaming1,semistreaming2} (sometimes just called the \defn{graph-streaming} model), an algorithm is presented with a \defn{stream} $\graphstream$ of updates (each an \defn{edge insertion} or \defn{deletion}) where the length of the stream is $\streamlength$. 
Stream $\graphstream$ defines an undirected input graph $\graph = (\nodes,\edges)$.  The challenge in this model is to compute (perhaps approximately) some property of $\graph$ given a single pass over $\graphstream$ and at most $o(\nodesize^2)$ (and ideally $O(\nodesize \polylog(\nodesize))$) words of memory. The model can be extended to hypergraphs as well and the formalism below applies to both settings.

Each stream update has the form $(e, \Delta, w(e))$ where $e = (u,v)$ for $u,v \in \nodes$, $u < v$, $\Delta \in \{-1,1\}$ where $1$ indicates an edge insertion and $-1$ indicates an edge deletion, and $w(e)$ denotes the weight of the edge. For most of the problems considered in this paper, the graph is unweighted, i.e., $w(e) = 1$ for all updates and in these cases we omit $w(e)$ in the update notation.  Let $\streamelement_i$ denote the $i$th element of $\graphstream$, and let $\graphstream_i$ denote the first $i$ elements of $\graphstream$.  
Let $\edges_i$ be the edge set defined by $\graphstream_i$, i.e., those edges which have been inserted and not subsequently deleted by step $i$.  The stream may only insert edge $e$ at time $i$ if $e \notin \edges_{i-1}$, and may only delete edge $e$ at time $i$ if $e \in \edges_{i-1}$. 

Once every update in the graph stream has been processed, a single query is performed, to which the algorithm returns the computed property of the graph. The query procedure is performed using the $O(\nodesize \polylog(\nodesize))$ memory retained by the algorithm at the conclusion of the stream.

\paragraph{Vertex-based Sketches.}
In this paper we present techniques that apply to a large family of graph sketch algorithms\cite{mcgregor2014graph,Ahn2012,KapralovLMMS13,GuhaMT15,assadi2023tight} called vertex-based sketches.
\begin{definition}[Vertex-based sketch\cite{GuhaMT15}]
\label{defn:vertex}
We say a linear measurement is local for vertex $v$ if the measurement only depends on edges incident to $v$, i.e., $ce = 0$ for all edges
that do not include $v$. We say a sketch is vertex-based if every linear measurement is local to some vertex.
\end{definition}

In other words, a vertex-based sketch is partitioned such that each part is mapped to a unique vertex in the graph, and each edge update only needs to be applied to the sketches associated with its endpoints.

\paragraph{External-memory model.}
In the \defn{external-memory (EM)} model~\cite{Vitter01}, memory is partitioned into RAM and disk. RAM has size $\memsize$ and disk has unbounded size. A word stored in RAM may be accessed at no cost, while disk is accessed in blocks of $\blocksize = o(\memsize)$ words. We refer to a disk block read or write as an I/O, and the goal is to minimize the number of I/Os required for an algorithm.

We will use the shorthand notation $\scan(N) = \Theta\left(\frac{N}{B}\right)$, $\sort(N) = \Theta\left(\frac{N}{B}\log_{M/B}\left(\frac{N}{B}\right)\right)$, and $\permute(N) = \min(N,\sort(N))$ for the optimal I/Os to scan, sort, and permute data of size $N$ in external memory, respectively~\cite{ChiangGoGr95}.

\paragraph{External semi-streaming model.}
For convenience, we restate the definition of our external semi-streaming model here.

In the \defn{external semi-streaming model}, edge insertions or deletions arrive in a stream. An algorithm in this model has $\memsize = \Omega(\polylog(\nodesize))=o(\nodesize)$ RAM available, and $\disksize = O(\nodesize \polylog(\nodesize)) = o(\nodesize^2)$ disk space. 
A word in RAM is accessed at no cost, and disk is accessed in blocks of $\blocksize = o(\memsize)$ words at a cost of a single I/O. 

 \section{Detailed Discussion of Results}
\label{sec:results}
\begin{table*}[h]
\footnotesize
  \begin{tabular}{l|l|l}
    \toprule
    \textbf{Algorithm} & \textbf{Vertex Sketch Size} & \textbf{I/O Cost} \\
    \midrule
    Connected Comp. & $\phi_1 = O(\log^2\nodesize)$ & $O\left(\permute(\streamlength) + \sort(\nodesize\phi_1)\right)$\\
    
    Bipartiteness Testing & $\phi_2 = O(\log^2\nodesize)$ & $O(\permute(\streamlength) + \sort(\nodesize\phi_2))$\\
    
    Hypergraph Conn. & $\phi_3 = O(r^2\log^2\nodesize)$ & $O(\permute(r\streamlength) + \scan(\streamlength\phi_3 / \memsize) + \sort(\nodesize) \cdot \poly(r,\log\nodesize))$ \\
    
    $\epsilon$-Appx MST Weight & $\phi_4 = O(\epsilon^{-1}\log^2\nodesize)$ & $O(\permute(\streamlength) + \scan(\streamlength\phi_4 / \memsize) + \sort(\nodesize) \cdot \poly(\epsilon^{-1},\log\nodesize))$ \\
    
    $k$-Connectivity  & $\phi_5 = O(k\log^2\nodesize)$ & $O(\permute(\streamlength) + \scan(\streamlength\phi_5 / \memsize) + \sort(\nodesize) \cdot \poly(k,\log\nodesize))$ \\
    
    $\epsilon$-Appx Min Cut & $\phi_6 = O(\epsilon^{-2}\log^4\nodesize)$ & $O(\permute(\streamlength) + \scan(\streamlength\phi_6 / \memsize) + \sort(\nodesize) \cdot \poly(\epsilon^{-1},\log\nodesize))$ \\
    
    $\epsilon$-Cut Sparsifier & $\phi_7 = O(\epsilon^{-2}\log^5\nodesize)$ & $O(\permute(\streamlength) + \scan(\streamlength\phi_7 / \memsize) + \scan(\nodesize^{2+o(1)}) \cdot \poly(\epsilon^{-1},\log\nodesize))$ \\
    
    $2(1+\epsilon)$-Appx & $\phi_8 = O(\epsilon^{-2}\log^2\nodesize)$ & $O(\permute(\streamlength) + \scan(\streamlength\phi_8 / \memsize) + \sort(\nodesize^2) \cdot \poly(\epsilon^{-1},\log\nodesize))$\\

    Densest Subgraph &&\\
    
  \bottomrule
\end{tabular}
\caption{Space and I/O bounds for our algorithms in words of space. $\streamlength$ denotes the length of the stream, $\nodesize$ denotes the number of vertices in the graph, and $\memsize$ denotes the size of RAM. To improve readability, we report the space of our algorithms in terms of $\phi$, the size of a vertex sketch (all of the reported sketch algorithms except densest subgraph are vertex-based; see Section~\ref{sec:easy_algs}). The total space for algorithm $i$ is $O(\nodesize\phi_i)$. For hypergraph connectivity, $r$ denotes the maximum hyperedge cardinality.}
\label{tab:results}
\vspace{-1em}
\end{table*}

In Table~\ref{tab:results}, we summarize the space and I/O bounds for the \modelname algorithms we present in this paper. All of our sketches are vertex-based (see Definition~\ref{defn:vertex}) so for each algorithm $A_i$ its sketch is partitioned into exactly $\nodesize$ equal-sized \defn{vertex sketches} and we denote the size of a vertex sketch as $\phi_i$. For most of these algorithms, the I/O cost is dominated by the cost of permuting the input stream, subject to mild assumptions. Specifically, this is true when $N$, the length of the stream, is greater than the size of the sketches, and $\memsize = \Omega(\phi_i)$, i.e., a single vertex sketch fits in RAM.

\begin{table*}[h]
\footnotesize
  \begin{tabular}{l|l|l|l|l}
    \toprule
    \textbf{Algorithm} & \multicolumn{2}{c|}{\textbf{Sketch}} & \multicolumn{2}{c}{\textbf{Ext. Mem.}}  \\
    \cline{2-5}
    ~ & I/O & Space & I/O & Space \\

    \midrule
    Connected Components & \green{Better} & \yellow{Same} & \yellow{Same} & \green{Better} \\
    
    Bipartiteness Testing & \green{Better} & \yellow{Same} & \yellow{Same} & \green{Better} \\
    
    Hypergraph Connectivity & \green{Better} & \yellow{Same} & \green{First} & \green{First} \\
    
    $\epsilon$-Approximate MST Weight* & \green{Better} & \yellow{Same} & \red{Worse} & \green{Better} \\
    
    $k$-Connectivity  & \green{Better} & \yellow{Same} & \yellow{Conditional} & \green{Better} \\
    
    $\epsilon$-Approximate Minimum Cut &  \green{Better}  & \red{$O(\log)$-factor worse} & \yellow{Conditional} & \green{Better}  \\
    
    $\epsilon$-Cut Sparsifier & \green{Better} & \yellow{Same} & \green{First} & \green{First} \\
    
    $2(1+\epsilon)$-Approximate Densest Subgraph* & \green{Better} & \yellow{Same} & \green{First} & \green{First} \\
    \bottomrule
\end{tabular}
\caption{Comparison of our \modelname algorithms' space and I/O complexities to the best existing graph sketching and external-memory algorithms. For example, ``\green{Better}'' indicates that the \modelname algorithm has a lower cost than the other algorithm, and ``\red{Worse}'' indicates that it has a higher cost.
Note for MST weight we compute an approximation while the best EM algorithm solves it exactly, and for densest subgraph we compute a $2(1+\varepsilon)$-approximation while the existing sketch gives a $(1+\varepsilon)$-approximation.
}
\vspace{-1em}
\label{tab:comp_sum}
\end{table*}
 
Table~\ref{tab:comp_sum} compares these bounds to those of the best existing graph sketch and external-memory algorithms for the problems we study. To compare I/O costs against external memory graph algorithms which assume a static input graph, we treat them as having an insert-only input stream of length $\streamlength = \edgesize$. 
For reference, the full details of the space and I/O costs of these existing algorithms are summarized in Table~\ref{tab:standard} in Appendix~\ref{app:standard}. 

\paragraph{Comparison to existing graph sketches.} Our \modelname graph sketches always have significantly lower I/O costs than existing graph sketches for the same problems, and always match their space costs (with the exception of the cut sparsifier sketch, which uses a $\log\nodesize$ factor more space). 

\paragraph{Comparison to existing external-memory algorithms.} Our \modelname graph sketches always use less space than existing external-memory algorithms except when $\streamlength = o(\nodesize\phi_i)$, i.e., when the graph is very sparse and the input stream is very short. Our algorithms for hypergraph connectivity, approximate densest subgraph, and cut sparsification are the first non-trivial external-memory algorithms for these problems to our knowledge. For connected components and bipartiteness testing, our sketches essentially match ($\permute(\streamlength)$ vs. $\sort(\streamlength)$) the I/O costs of the best known algorithms. For approximate MST, our graph sketch has worse I/O performance than the best exact EM algorithm. For $k-$connectivity, our algorithm performs better than the best EM algorithm when $k = O(\min(\memsize \log^3\nodesize, \edgesize/\nodesize))$. For $\varepsilon-$approximate min cut, 
our algorithm performs better than the best (exact) EM algorithm if $\epsilon=\Omega\left(\max\left(\memsize^{-1/2} \log^{-1/2}\nodesize, (\nodesize/\edgesize)^{1/4}\loglog^{1/4}\nodesize) \right)\right)$.

The upshot is that many of our \modelname algorithms have I/O costs comparable to or better than existing external-memory algorithms. 

\section{An \Modelname Algorithm for Connectivity}
We begin by considering the problem of computing the connected components problem in the \modelname model. Ahn \etal~\cite{Ahn2012} give a $O(\nodesize \log^2 \nodesize)$-space sketching algorithm to solve semi-streaming connectivity and, later, Nelson and Yu~\cite{nelson2019optimal} prove that this space cost is optimal. Subsequent work by Tench \etal~\cite{tenchwevz22gz} presents a somewhat I/O-efficient version of the connectivity sketch, which was sufficient to achieve good performance when implemented, but falls short of our desired properties for a \modelname algorithm. In this section we present a sketching algorithm which improves on the I/O cost of the Tench \etal algorithm, and actually matches the I/O cost of the best known external-memory connected components algorithm (assuming that $\edgesize = \Omega(\nodesize\log^2\nodesize)$, that is, when the graph is not so sparse that sketching would save no space). Further, we show a lower bound in Section~\ref{sec:easy_algs} for a large family of sketch algorithms that implies as a special case that our connectivity sketch is I/O-optimal.

\subsection{Ahn \etal's Connectivity Sketch}
We begin by reviewing the dynamic semi-streaming connectivity algorithm of Ahn, Guha, and McGregor, which we refer to as \algname ~\cite{Ahn2012}:
\begin{theorem}
There exists an $O(\nodesize \log^2(\nodesize))$-space dynamic streaming algorithm for the connected components problem that succeeds with high probability (w.h.p.) in $\nodesize$.
\end{theorem}
The algorithm in the above theorem is a linear sketching algorithm:

\begin{definition}
A linear measurement of a graph on n vertices is defined by a set
of coefficients $\{ce : e \in {\binom{\nodesize}{2}}\}$. Given a
graph $\graph = (\nodes, \edges)$, the evaluation of this measurement is defined as $\sum_{e\in\edges}{ce}$. A sketch is a collection of (non-adaptive) linear measurements. The cardinality of this collection is referred to as the size of the sketch. We will assume that the magnitude of the coefficients $ce$ is $\poly(n)$.
\end{definition}

Nearly all known small-space algorithms for data stream problems whose input streams have both insertions and deletions are linear sketch algorithms. Further, Li \etal ~\cite{mightaswell} show that the family of linear sketch algorithms are essentially universal for insert/delete data stream problems: for any space-optimal algorithm that succeeds with constant probability, there is an equivalent linear sketching algorithm that uses at most a logarithmic factor more space than optimal.

In this paper we elide many details of \algname but make several necessary observations here. Let $\sketch(\graph)$ denote a connectivity sketch of graph $\graph$. $\sketch(\graph)$ can be partitioned into $\nodesize$ $O(\log^2(\nodesize))$-size data structures $\sketch^0(\graph), \sketch^1(\graph), \dots \sketch^{\nodesize-1}(\graph)$ which have the property that edge update $e = (u,v,\Delta)$ induces changes only to $\sketch^u(\graph)$ and $\sketch^v(\graph)$. We call $\sketch^u(\graph)$ the \defn{vertex sketch of vertex $u$}.
For a subset $A \subseteq V$, we let $\sketch^A(\graph) = \bigcup_{u \in A} \sketch^u(\graph)$ denote the union of the sketches of the vertices in $A$.

Crucially, the sketch is linear; i.e., it has the property that $\sketch^u(\graph) = \sum_{(v,w) \in \edges} \sketch^u((v,w))$ for all $u \in \nodes$. $\sketch(\graph)$ is computed by keeping a running sum of the vertex sketches of each stream update.

$\sketch^u(\graph)$ may be queried to sample edges from the neighborhood of $u$, and $\sketch^A(\graph)$ may be queried to sample edges from the cut $(A, \nodes \setminus A)$. After the stream, the algorithm finds the connected components using Boruvka's algorithm~\cite{boruvka}, querying sketches to sample edges leaving each component.

Some algorithms in this paper make use of multiple connectivity sketches of the same graph, where each connectivity sketch is initialized with different random bits. We denote the $i$th connectivity sketch as $\sketch_i(\graph)$.

In the \modelname model, \algname has high I/O complexity. While Tench \etal's modified algorithm gets better performance in practice, their I/O complexity can be improved. See Appendix~\ref{app:iobad} for a detailed discussion.

\subsection{Connectivity: Sketching the Input Stream}\label{subsec:guttertree}
We present a new algorithm called \externalcc  that computes the connected components of the graph defined by the input stream. 
We refine the stream-sketching technique of Tench \etal ~\cite{tenchwevz22gz} and introduce a new query procedure. As a result, \externalcc uses fewer I/Os to sketch the input stream and to compute connectivity from the sketch than either \algname or Tench \etal's algorithm.

First we describe how \externalcc sketches the input stream.

As stream updates come in, we process them in batches. For each batch, the updates are initially sent to disk after collecting $B$ at a time. Once $O(\nodesize \log^2\nodesize)$ updates have been collected, we empty the batch by applying all its updates to the corresponding vertex sketches. We repeat this process for every succeeding $O(\nodesize \log^2\nodesize)$ updates until the stream terminates.

We now describe the batching and update procedure in more detail. We arbitrarily partition the vertices of the graph into \defn{vertex groups} of cardinality $\max\{1, \blocksize/\log^2(\nodesize)\}$.  
Let $\nodegroup \subset \nodes$ denote a vertex group.
We store $\sketch^\nodegroup(\graph)$, the vertex sketches associated with the vertices in $\nodegroup$, contiguously on disk.  This allows $\sketch^\nodegroup(\graph)$ to be read into memory I/O efficiently: if vertex groups are of cardinality 1, then $\blocksize$ is smaller than the size of a vertex sketch, and if each vertex group has cardinality $\blocksize/\log^2(\nodesize) > 1$, then the sketches for the group have total size $\Theta(\blocksize)$.

For each vertex group, we have a corresponding \defn{update buffer}, which will collect the updates affecting that vertex group so that they can be processed efficiently.
The update buffers are stored on disk in the same order as the vertex sketches.
Since each edge update $e = (u,v)$ needs to be applied to both endpoints, before performing the following update procedure, we make a copy of each edge update, and mark one copy with the left endpoint $u$ and one with the right endpoint $v$. 
Once a batch is full, we permute the updates into the update buffers corresponding to the marked (left or right) endpoints.
If the update buffer belonging to any vertex group $\nodegroup$ fills up during the course of the permuting procedure, we immediately empty it by reading $\sketch^\nodegroup(\graph)$ into memory and applying the updates in the buffer.
This ensures that all updates can be placed in their target update buffers without overflow.
Once all elements have been permuted, we simultaneously scan through the update buffers and the sketches, applying each remaining update to the corresponding sketch.

\begin{lemma}
\label{lemma:ingest}
\externalcc's stream ingestion uses $O(\nodesize \log^2(\nodesize))$ space and 
$O  \left( \min \left( \streamlength, \frac{\streamlength}{B}\log_{M/B}((\nodesize / B) \log^2\nodesize) \right) + \scan(\nodesize \log^2 \nodesize) \right)$ I/Os.
\end{lemma}

\begin{proof}
\externalcc's sketch data structures and update buffers use $O(\nodesize \log^2(\nodesize))$ space. We also store at most $\nodesize \log^2 \nodesize$ pending updates from a batch on disk at a time. 
    
Now we analyze the I/O cost of \externalcc's ingestion procedure.
Let $\mathcal{B}$ be a batch of updates. Collecting the updates for $\mathcal{B}$ and writing them to disk costs $\scan(\nodesize\log^2 \nodesize)$ I/Os in total. 
Permuting the updates costs $\permute(\nodesize \log^2\nodesize) = \min(\nodesize \log^2\nodesize, \sort(\nodesize \log^2\nodesize))$ I/Os. 
Applying the updates to the sketches costs $O(\scan(\nodesize\log^2 \nodesize))$ I/Os, which is dominated by the permute cost unless $\streamlength < \nodesize \log^2 \nodesize$.

There are $\frac{\streamlength}{2\nodesize\log^2\nodesize}$ batches, so the overall I/O complexity of ingesting the stream is $O \left( \frac{\streamlength}{\nodesize\log^2\nodesize} \permute(\nodesize \log^2\nodesize)\right) = O \left( \min \left( N, \frac{N}{B}\log_{M/B}( (\nodesize / B) \log^2 \nodesize) \right) \right)$.

Finally, we have to address the case where $\streamlength < \nodesize \log^2 \nodesize$. In this case, still need to scan through the sketches after the permute operation, so the minimum I/O cost is $\scan(\nodesize \log^2 \nodesize)$.
\end{proof}

\subsection{Extracting the Components from the Sketch}
We now describe \externalcc's procedure for computing connected components once all stream updates have been processed.

Our algorithm proceeds through $O(\log \nodesize)$ rounds each consisting of three phases.
In the first phase, an edge is recovered from the sketch of each current connected component. 
These edges make up a `merge list'. 

In the second phase, for each edge in the merge list, its endpoints are merged in a union--find data structure which keeps track of the current connected components.
We use the union--find data structure of Agarwal \etal~\cite{union_find} for efficient batched computation.
To obtain the best bounds from this data structure, we need the merge list to be free of redundant merges (i.e., two different edges  that effectively merge the same components). 
To achieve this, we preprocess the merge list as follows.
First, we find the supernodes corresponding to the endpoints of each edge in parallel: we sort edges by increasing node ID of their left endpoints.
Then we simultaneously scan through this sorted edge list and the union--find data structure to get the parent of each left endpoint. 
Repeating this $\alpha(\nodesize)$ times gives the component of each left endpoint, where $\alpha$ denotes the inverse Ackermann function.
Finally, we repeat these steps for the right endpoints of each edge.
At the end of this phase we have a new list of $O(\nodesize)$ merges of the form $u \to v$, indicating that the sketch for component $u$ should be merged into the sketch for component $v$. 
In order to remove redundant merges, we construct a graph $\cert$ from this list, where each vertex corresponds to a supernode in the list, and an edge $(u,v)$ for every merge $u \to v$. 
Connected components in $\cert$ correspond to nodes that will all be merged together when all merges in the list are completed. 
Therefore, we run the external-memory connected components algorithm of Chiang \etal~\cite{ChiangGoGr95} to compute these connected components. 
We then replace the merges in the list with merges of the form $u \to v'$, where $v'$ is the representative of the connected component of $v$. 
Finally, we run this list of merges through the union find data structure. 

In the third phase, for each pair of connected components merged in phase 2, the corresponding sketches are summed. 
Summing the sketches of the merged components together na\"\i{}vely is I/O efficient if $\blocksize = O(\log^2(\nodesize))$, since the disk reads and writes necessary for summing sketches are the size of a block or larger. 

If $\blocksize = \omega(\log^2(\nodesize))$, that is if sketches are much smaller than the block size, then we need a more sophisticated merge procedure. Since the merges performed in each round of \boruvka are a function both of the input stream and of the randomness of the sketches, these merges induce random accesses to the sketches on disk if performed directly. In this case, this induces $O(1)$ I/Os per sketch merged for a total cost of $O(V)$ I/Os to perform all the sketch merges.
However, this operation can be done more efficiently by sorting the merge list by merge source in node ID order 
and sorting the sketches in the same order.  
We then scan through the sketches, marking each sketch with its merge destination, and finally sort the sketches by these merge destinations.
Now, because the sketches for each component are stored contiguously, we can perform all the merges with one more scan of the sketches.

\begin{lemma}
\label{lemma:query}
Once all stream updates have been processed, \externalcc computes connected components using 
$O(\sort(\nodesize\log^2(\nodesize)))$ I/Os.
\end{lemma} 

\begin{proof}
We analyze the I/O cost of each phase of a round separately, as the total cost is a constant factor greater than the cost of performing the first round. This is because the number of components we operate on, and thus the amount of I/Os we perform, decreases by a constant factor in each round.

To query the sketches in the first phase, each sketch must be read into RAM, queried, and the result appended to a list. This can be done with a single scan using $O(\scan(\nodesize \log^2(\nodesize)))$ I/Os to produce a list of $O(\nodesize)$ edges.

Finding the supernodes corresponding to each edge in the merge list consists of a scan and a sort of data of size $O(V)$ per iteration, with $\alpha(\nodesize)$ iterations. This gives $O(\sort(\nodesize) \alpha(\nodesize)$ I/Os in total to find the supernodes. 
Since the graph $\cert$ has $O(\nodesize)$ edges, computing its connected components takes $O(\sort(\nodesize))$ I/Os \cite{ChiangGoGr95}. Performing the queries in the merge list of size $O(\nodesize)$ in the union--find data structure takes $O(\sort(\nodesize))$ I/Os\cite{union_find}.

In the third phase, if sketches are larger than the block size, the cost of summing the sketches is the scan cost of $O(\nodesize \log^2(\nodesize)/\blocksize)$.

If sketches are smaller than blocks, we sort the merge list in $O(\sort(\nodesize))$ I/Os, then sort the sketches in the same order in $O(\sort(\nodesize\log^2(\nodesize)))$ I/Os. The simultaneous scan through the sketches and merge list is a low order cost to the sort operations. Finally, sorting the sketches by their marked merge destinations costs $\scan(\nodesize\log^2 \nodesize) + \sort(\nodesize) = \scan(\nodesize\log^2\nodesize)$. I/Os, after which this round is completed with one additional scan.

The cost to run our entire query only a constant factor greater than the cost of these three phases or $O(\scan(\nodesize \log^2 \nodesize)) + O(\sort(\nodesize)\alpha(\nodesize)) + O(\sort(\nodesize \log^2\nodesize)) = O(\sort(\nodesize \log^2(\nodesize))$.
\end{proof}

\externalcc can also be used as an external-memory connected components algorithm on a static graph. Provided the graph has enough edges, it matches the best known upper bound of $\sort(\edgesize)$ I/Os for connected components~\cite{ChiangGoGr95}, and uses $\softO(V)$ less space.

\begin{corollary}
When $\edgesize = \Omega(\nodesize \log^2(\nodesize))$, \externalcc solves the connected components problem in 
$O  \left( \min \left( \edgesize, \frac{\edgesize}{B}\log_{M/B}( (\nodesize/B) \log^2\nodesize) \right) \right)$ I/Os.
\end{corollary}

\section{A General Transformation for \Modelname Graph Sketching}
\label{sec:easy_algs}

In the prior section we showed that connectivity can be solved via sketching for essentially the same I/O cost as the best known (non-sketching) external-memory algorithm. This surprising fact is due in part to the fact that it is a vertex-based sketch, because each edge update only needs to be applied to the sketches associated with its endpoints. 
In this section we show how to, for any such algorithm, sketch the input stream using essentially the number of I/Os required to permute the input stream --- without increasing the space cost. We then show that this I/O cost is optimal or nearly optimal (depending on the problem) via a lower bound based on a reduction to sparse matrix/dense vector multiplication.

\begin{theorem}
\label{thm:ingest}
For any single-pass vertex-based sketch streaming algorithm, 
the input stream can be processed using 
\begin{align*} \vsketch(&\streamlength, \nodesize, \phi) \coloneqq \\& O  \left(  \min \left( \streamlength, \frac{\streamlength}{B}\log_{\memsize/B}(\nodesize\phi/B) \right)\right. \\& +  \left.\scan(\streamlength \phi/\memsize) + \scan(\nodesize \phi) \right) \end{align*}
I/Os and total space $O(\nodesize \phi)$. \end{theorem}

\begin{proof}
We follow the procedure described in Section \ref{subsec:guttertree}, except we set the size of a batch to be $O(\nodesize\phi)$, vertex groups to have cardinality $\max\{1, \blocksize/\phi\}$ and the size of each update buffer to be $\phi$.

In the case where $\phi = o(\memsize)$, the result follows immediately. 

In the case where $\phi = \Omega(\memsize)$, the procedure for applying the updates in an update buffer to a sketch is more complicated and expensive. (In this case, vertex groups will always be of size 1.)

Sorting each batch costs $\permute(\nodesize \phi) = \min(\nodesize \phi, \sort(\nodesize \phi))$ I/Os. 
Whenever the update buffer for vertex $u$ fills, we apply the updates by holding the first $O(\memsize)$ elements of the buffer in memory and scanning over $\sketch_u(\graph)$ in $O(\memsize)$-size chunks, applying the updates to each chunk. 
This costs $\scan(\phi)$ I/Os per $O(\memsize)$ updates, and there are at most $\nodesize \phi / \memsize$ such sets of updates per batch, for a total I/O cost to apply the updates of $O(\scan(\nodesize \phi^2 / \memsize))$ per batch.

There are $\streamlength / (V\phi)$ batches, for a total ingestion I/O cost of 
$O \left( \frac{\streamlength}{\nodesize \phi} (\permute(\nodesize \phi) + \scan(\nodesize \phi^2 / \memsize)) \right) 
= O \left( \min \left( \streamlength , \streamlength \log_{\memsize/B} (\nodesize \phi/B) \right) + \scan(\streamlength \phi/\memsize)  \right)$.

Finally, similarly to Lemma~\ref{lemma:ingest}, we have an I/O cost of $\scan(\nodesize \phi)$ in the case that  $\streamlength < \nodesize \phi$. 
\end{proof}

\subsection{A Matching I/O Lower Bound}
Now we show an I/O lower bound for sketching the input stream for any vertex-based sketch algorithm. Depending on the problem and $\memsize$, the lower bound either matches the upper bound exactly or has a $O(\log\nodesize)$ gap.

To argue for a lower bound, it is useful to separate the sketching from the data-structural aspect, so that we can argue about the I/O complexity of the data-structure.
Here, we focus on sketching algorithms such that
\begin{itemize}
  \item The sketch is created from $\streamlength$ edge updates (insert or delete) that are presented in arbitrary order to the data structure. 
  \item Sketches work with vertex sketches of polylog many numbers that are treated as atoms of the I/O model. 
    These atoms can only be added up (using associativity and commutativity). 
  \item An edge contributes precisely to the two vertex sketches of its endpoints.  
  An edge is treated as an I/O atom. The I/O algorithm can read out the (polylog many) number atoms from an edge atom at no cost in internal memory (transformation). The difference between insert and delete is not visible to the I/O data structure---it merely adds up all the components of the sketches. This can easily be used to implement deletions by canceling contributions. The data structure is not required to check if the multiplicity of an edge is 1.
  The transformation function is available to the I/O algorithm at no cost.
  \item All vertex sketches have the following two-dimensional sparsity structure:  The sketch consists of numbers organized in $\rho$ rows and $\tau = \Theta(\log N)$ columns.  The contribution of a single edge satisfies
  \begin{itemize}
    \item The first column of each row always has a non-zero entry.
    \item Each row starts with non-zero entries followed by zero entries.
    \item The probability of a row having $k$ non-zero entries is $(1/2)^{k}$.
    This is truncated at $\tau$, if the row should be longer.
  \end{itemize}
\end{itemize}
Encapsulating the nature of hash functions defining the contributions of an edge to a sketch, the above ``magic expansion'' of an edge atom into many number atoms is justified: while it is deterministic, to the I/O algorithm everything looks as if it is completely random and has no structure that can be exploited for improved efficiency.

\begin{description}
  \item[Add edge $(u,v)$] The item is a single atom of the I/O model. 
  \item[Finalize] The algorithm produces a sorted list of vertex sketches, each having $\rho$ rows and $\tau$ columns.
\end{description}
The parameters of this interface are $N$, the number of add (insert/delete) operations; $\rho$, the number of rows in a vertex sketch; $\tau$, the number of columns in a vertex sketch; $\phi=\rho\tau$, the number of atoms in a vertex sketch; and $V$, the number of vertices.

The upper bound is
\begin{align*}
  O  \Big(  \min & \big( \streamlength, \frac{\streamlength}{B}\log_{\memsize/B}(\nodesize\phi/B) \big) 
    \\ &  + \scan(\streamlength \phi/\memsize) + \scan(\nodesize \phi) \Big) 
\end{align*}
\begin{theorem}\label{thm:IOlowerbound}
  Assume there is an implementation of the above interface with the above parameters.
  There exists a sequence of $N$ Add operations followed by a Finalize such that the following number of I/Os are necessary:
\[ 
  \Omega \left( \min \left( \streamlength, \frac{\streamlength}{B}\log_{\memsize/B}(\nodesize\phi/B) \right) 
  + \scan(\streamlength \rho /\memsize)\right)
\]
\end{theorem}

\begin{proof}
The first lower bound is justified by a reduction of sparse matrix multiply. 
Here we assume that the adversary can choose the numbers that are inserted into the sketches arbitrarily (by choosing an appropriate hash function).
Let $A$ be a sparse matrix with $2N$ columns and $\rho\nodesize$ rows, where each column has precisely one entry that is given in column major layout. 
Now, the above data structure can be used to compute $Ax$ by scanning over $x$ and $A$ in its layout and inserting items.  
Assume the entry is in row~$i$, calculate $k=i/ \rho$ and $h= i \mod \rho$.
Create an edge that involves vertex $k$ and let the entry of the vertex sketch in row~$h$ and in the first column be the only non-zero entry associated with this endpoint of the edge. 
Set the value of this entry to $a_{ij}x_j$.
Each consecutive pair of these are leading to an add operation.
The Finalize operation creates the sum of the basic sketches, from whom we can read out the vector $y$ in the first columns of the vertex sketches.
Hence, the lower bound of \cite{SpMxVTCS2010} applies and justifies the term.

The second line is a lower bound following from a volume consideration in Hong and Kung rounds:
It is well known \cite{HongKu81,SpMxVTCS2010} that there is no asymptotic penalty for assuming that an I/O program operates in rounds of $M/B$ I/Os, where first the memory is loaded, then computation happens, and finally the memory is written to disk.
This can simulate any I/O program with twice the memory size and at most twice the number of I/Os.
In this normalized setting, it is clear what we mean by tracing the original atoms (including making copies), and atoms that are part of the final output (including reduction/summation steps).
Obviously, these traces must meet for output atoms that depend on input items, and the ``transformation'' from an input item to an output atom must happen in some Hong and Kung round. 
In each such round, there are at most $M$ input items loaded, and at most $M$ output atoms are written out, so at most $M^2$ different useful transformations can happen per round. 
There is a need for a total of at least $2\streamlength\rho$ transformations, leading to a lower bound of $2\streamlength\cdot\rho/M^2$ Hong and Kung rounds, which is equivalent to a lower bound of 
$M/B\cdot\streamlength\cdot\rho/M^2 = \scan(\streamlength\rho/M)$ I/Os, as written in the second term.
\end{proof}

Observe that for the case of $M\ge\phi$ (and the scanning of the sketch not dominating the algorithm), we have asymptotically matching upper and lower bounds.
This is the case for the previously stated assumption $\memsize = \Omega(\polylog(\nodesize))$ of our hybrid graph streaming setting. 
Otherwise, in an extended parameter range of the I/O data structure,  with the results presented so far, there is a $\Theta(\log N)$ gap between upper and lower bounds. 
To almost close this gap, we can improve the upper bound by the following considerations.

\begin{lemma}\label{lem:gemometricMax}
  Assume there are $\rho$ random variables $X_i\ge 1$ with $p[X_i=j] = (1/2)^j$. 
  Then the probability $p[\exists i: X_i > j] \le \min(1,\rho(1/2)^{j})$
\end{lemma}

\begin{proof}
  By a union bound.
\end{proof}

Now we split the data structure for the sketches by columns, $\tau'\ge 1$ columns per data structure, chosen such that  $\tau'\rho\le M$ (if possible, otherwise set $\tau'=1$). 
An edge is always inserted into the data structure for the first columns, and also in all the data structures where one of its rows has a  non-zero entry in one of the columns of the data structure.
Hence, the expected contribution of an edge to the input stream of the $k$-th data structures is 
$\min(1,\rho(1/2)^{k\tau'})$, i.e., 1 for $k\tau'<\log \rho$ and then geometrically decreasing.
This improves the $\scan(\frac{\streamlength\phi}{\memsize})$ term in the upper bound to $\scan(\frac{\streamlength\rho\log\log N}{\memsize})$.
As long as $\rho$ is polylogarithmic in $N$ this leads to a gap of $O(\log\log N)$ between upper and lower bound. 
If $M = \Omega(\rho\log\log N)$, there is no asymptotic gap between upper and lower bound.

\subsection{More \Modelname Algorithms}
Theorem~\ref{thm:ingest} immediately implies efficient \modelname algorithms for hypergraph connectivity (for bounded hyperedge cardinality $r$), bipartiteness testing, and $(1+\epsilon)$-approximating MST weight, all of which use $O(\nodesize\poly(\log(\nodesize),\epsilon^{-1},k,r))$ space and $\vsketch(\streamlength,\nodesize,\poly(\log(\nodesize),\epsilon^{-1},k,r))$ I/Os.

\paragraph{Hypergraph connectivity.}
In followup work, Guha \etal ~\cite{GuhaMT15} show that by using a slightly different vector encoding, their connectivity result can be extended to hypergraphs at the cost of a multiplicative $O(r^2)$ increase in the size of the sketch, where $r$ is the maximum hyperedge cardinality. The remainder of the algorithm is essentially unchanged.

\begin{corollary}
\label{cor:hypercon}
Given a hypergraph $\graph = (\nodes, \edges)$ with maximum edge cardinality $r$, there exists a 
$O(\vsketch(r\streamlength, \nodesize, r^2 \log^2 \nodesize) + \sort(r^2\nodesize\log^2\nodesize))$
-I/O algorithm which computes a spanning forest of $\graph$ w.h.p. and uses $O(r^2\nodesize\log^2(\nodesize))$ space.
\end{corollary}

\begin{proof}
The algorithm is only slightly different than the algorithm for connected components on (non-hyper-) graphs.
There are $O(\nodesize^r)$ possible edges in a hypergraph with maximum edge cardinality $r$, so we may encode all of them in a characteristic vector with length $\nodesize^r$. Then $\phi = r^2 \log^2 \nodesize$.~\cite{GuhaMT15}. 
Since each update needs to be applied to at most $r$ vertex sketches, we make at most $r$ copies of each update, and apply the ingestion procedure from Theorem~\ref{thm:ingest} to a stream that is now of length $r \streamlength$. 
This gives an ingestion cost of $\vsketch(r\streamlength, \nodesize, r^2 \log^2 \nodesize)$. 
The total space for the sketch is $O(r^2 \nodesize \log^2 \nodesize).$

The procedure for computing a spanning hypergraph from the sketch again is similar to connectivity on graphs. The first phase of a round is identical, requiring a scan over the sketches for $\scan(r^2 \nodesize \log^2(\nodesize))$ I/Os. Since a hypergraph edge may require up to $r - 1$ union--find updates, the cost of finding the components of each endpoint of each edge and merging them in the union find data structure is  $\sort(r\nodesize)\alpha(\nodesize))$. The mechanics of the third phase are unchanged, so the cost is $\sort(r^2 \nodesize \log^2(\nodesize))$. So the total cost of the first round is $O(\sort(r^2 \nodesize \log^2(\nodesize)))$. 
\end{proof}

To our knowledge, this is the first nontrivial external-memory algorithm for hypergraph connectivity.

\paragraph{Bipartiteness testing.}
We present an algorithm for testing whether a graph is bipartite.
The result is immediate: Ahn \etal~\cite{Ahn2012} reduce determining whether a graph $\graph = (\nodes, \edges)$ is bipartite to computing the number of connected components of a graph $D(\graph) = (\nodes', \edges')$ such that for each $v \in \nodes$ we add $v_1,v_2\in \nodes'$ and for each edge $(u, v) \in \edges$, we add two edges $(u_1, v_2)$ and $(u_2, v_1)$. This, combined with Lemmas ~\ref{lemma:ingest} and ~\ref{lemma:query}, give the following theorem:

\begin{theorem}
There exists a $O(\nodesize\log^2(\nodesize))$-space, 
$\vsketch(\streamlength, \nodesize, \log^2 \nodesize)$-
I/O algorithm for bipartiteness testing that succeeds w.h.p.
\end{theorem}

\paragraph{Approximating minimum spanning tree weight.}
Ahn \etal~\cite{Ahn2012} show how to $(1+\epsilon)$ approximate the weight of the minimum spanning tree (MST) of a graph $\graph = (\nodes, \edges)$ by using their connected components (CC) sketches. For edge weights in the range $[W]$, they create $r = \log_{1+\epsilon}(W)$ CC sketches and use the $i$th to sketch $\graph_i = (\nodes,\edges_i)$, where $\edges_i = \{e\in\edges: w(e) \leq (1+\epsilon)^i\}$, and $w(e)$ denotes the weight of edge $e$. They prove that 
$$w(T) \leq n - (1 + \epsilon)^r + \sum_{i=0}^r{\sigma_i cc(\graph_i)} \leq (1 + \epsilon)w(T)$$

where $T$ denotes the minimum spanning tree of $G$, $cc(\graph_i)$ denotes the number of connected components in $G_i$, and $\sigma_i = (1 + \epsilon)^i - (1 + \epsilon)^{i-1}$.

It suffices to find the number of connected components of each $G_i$. Constructing the $O(\log(\nodesize))$ CC sketches via Theorem~\ref{thm:ingest} uses $O(\vsketch(\streamlength, \nodesize, \varepsilon^{-1}\log^2\nodesize))$ I/Os, and reconstructing the spanning forests from each sketch takes $\sort(\nodesize \log^2(\nodesize))$ I/Os by Lemma ~\ref{lemma:query}.
This gives the following theorem:

\begin{theorem}
There exists a $O(\epsilon^{-1} \nodesize\log^2(\nodesize))$-space, 
$O(\vsketch(\streamlength, \nodesize, \varepsilon^{-1}\log^2\nodesize) + \varepsilon^{-1}\sort(\nodesize \log^2(\nodesize)))$ I/O algorithm which $(1 + \epsilon)$-approximates minimum spanning tree weight w.h.p.
\end{theorem}

\section{More Extraction Techniques for \Modelname Algorithms}

The algorithms described in the previous section rely on both the general transformation described in Theorem ~\ref{thm:ingest} and the procedure described in Lemma ~\ref{lemma:query} that computes a spanning forest from the sketches after stream ingestion. In general, while Theorem \ref{thm:ingest} provides a way to perform stream ingestion I/O efficiently on any single-pass vertex-based sketch algorithm, the challenge of how to minimize the extraction cost remains open. Most graph sketch algorithms have a post-stream procedure for \emph{querying} their sketch data structures to produce a sparse graph which retains (perhaps approximately) some property of the graph defined by the stream. We now present some \modelname algorithms that demonstrate how to minimize both the I/O cost of querying the sketch to produce a sparsifier, and then computing the answer from that sparsifier.

\paragraph{Testing k-edge-connectivity.}
We consider the problem of testing k-connectivity of a graph $\graph = (\nodes, \edges)$, or equivalently, exactly computing the minimum cut $\lambda(\graph)$ if $\lambda(\graph) \leq k$. 

We make use of the solution of 
Ahn \etal~\cite{Ahn2012},
which
 constructs a k-connectivity certificate $H = \bigcup_{i \in [k]} F_i $ where $F_0, F_1, \dots, F_{k-1}$ are edge-disjoint spanning forests of $\graph$. $H$ has the property that it is $k'$-edge connected iff $G$ is $k'$-edge connected for all $k' \leq k$. They find each $F_i$ by computing a connectivity sketch $\sketch_i (\graph \setminus \bigcup_{j<i}{F_j})$. This is done in a single pass over the stream: during the stream, they keep $k$ different sketches of $\graph$: $\sketch_0(\graph), \sketch_1(\graph), \dots, \sketch_{k-1}(\graph)$. We use $\kconsketch(\graph)$ to denote the concatenation of these $k$ connectivity sketches. After the stream, $\sketch_0(\graph)$ is used to find $F_0$ and the edges of $F_0$ are deleted from the remaining $k-1$ connectivity sketches. $\sketch_1(\graph \setminus F_0)$ can now be used to get $F_1$, whose edges are subsequently deleted from the remaining $k-2$ sketches and so on.

The extraction step of the above algorithm must access $\kconsketch(\graph)$ at least once so it has a trivial lower I/O bound of $\scan(k\nodesize\log^2\nodesize)$. Performing the spanning forest deletions naively requires a $k$ factor I/O overhead: $\scan(k^2\nodesize\log^2\nodesize)$. Recall that we would like our algorithm to have an extraction cost not much higher than the cost of computing the property on a sparse graph in external memory, so we must reduce this overhead. The primary challenge to doing so is that the edges found in each spanning forest induce deletions to subsequent spanning forests, which necessitate many random accesses with different deadlines. We solve this issue by scheduling deletions for each spanning forest $F_i$ carefully so that the deletion cost can be amortized over the cost of querying later sketches $S_j \forall j > i$. This reduces the I/O overhead from $k$ to $\log k$. Finally, we apply an exact min cut algorithm due to Geissman and Gianinazzi~\cite{extmincut} to compute the minimum cut of the union of the $k$ forests.

This gives the following theorem:

\begin{theorem}
\label{thm:kconnect}
There exists an $O(k\nodesize\log^2(\nodesize))$-space, 
$O(\vsketch(\streamlength, \nodesize, k\log^2 \nodesize) 
+ \ksketch(\nodesize, k)
+ \sort(k\nodesize)\log^4(\nodesize))$
-I/O algorithm for testing k-edge connectivity that succeeds w.h.p., where 

\[
\ksketch(\nodesize, k) =
\begin{cases}
\scan&(k \log(k)\nodesize  \log^2(\nodesize)), \\ &\text{when } k\log^2(\nodesize) =o(\memsize)\, \\
\scan&(k^2  \nodesize\log^2(\nodesize)),  \text{otherwise.}\\
\end{cases}
\]
\end{theorem}

\begin{proof}   
Constructing $\kconsketch(\graph) = \sketch_0(\graph), \sketch_1(\graph), \dots \sketch_{k-1}(\graph)$ takes $O(\vsketch(\streamlength, \nodesize, k \log^2 \nodesize)$ I/Os by Theorem ~\ref{thm:ingest}. 
We will now analyze the cost of forming $H$ from the sketches, and computing the minimum cut of $H$, in more detail.

Computing $F_i$ from $\sketch_i(\graph \setminus \bigcup_{j<i} F_j)$ takes $\sort(\nodesize\log^2(\nodesize))$ I/Os for each $0 \leq i < k$ by Lemma ~\ref{lemma:query}. When computing $F_0$, we already have $\sketch_0(\graph)$, so this is the only cost. Deleting the edges in $F_0$ from $\sketch_i(\graph)$ for all $0 < i < k$ takes $O(k\nodesize\log^2(\nodesize))/\blocksize)$ I/Os. Repeating this process to create each $\sketch_i(\graph \setminus \bigcup_{j<i} F_j)$ therefore takes $O(k^2\nodesize\log^2(\nodesize))/\blocksize)$ I/Os, but a more careful method allows us to reduce this cost when $k$ is small enough that a vertex sketch fits in RAM. For the remainder of the proof, we assume $\phi = k\log^2(\nodesize) = o(\memsize)$.

Let us focus on the cost of forming $\sketch_i(\graph \setminus \bigcup_{j<i} F_j)$ for some $i$. Since we begin knowing $\sketch_i(\graph)$, this cost is exactly the cost of deleting $\bigcup_{j<i} F_j$, the edges in the previous $i-1$ spanning forests, from $\sketch_i(\graph)$. 
For any $i < \log^2(\nodesize)$ the total cost to form $\sketch_i(\graph \setminus \bigcup_{j<i} F_j)$ is less than $\scan(\nodesize\log^2(\nodesize))$ because both $\sketch_i(\graph)$ and the list of deletions $\bigcup_{j<i} F_j$ have size $O(\nodesize\log^2(\nodesize))$.

Let $F_i' = \bigcup_{j = i\log^2(\nodesize)}^{(i+1)\log^2(\nodesize)-1} F_j$ for $0 \leq i < k/\log(\nodesize)$ and similarly let $\sketch_i'(\graph \setminus \bigcup_{j<i} F_j') = \bigcup_{j = i\log^2(\nodesize)}^{(i+1)\log^2(\nodesize)-1} \sketch_j(\graph \setminus \bigcup_{j'<j} F_{j'}')$. The cost of applying deletions to sketches while producing $F_0'$ is at most $\eta/2$ where $\eta = \scan(\nodesize\log^4(\nodesize))$ by Gauss summation. Similarly, the cost of the deletions for $F_1'$ is $3\eta/2$. Of this total, $\eta$ is from deleting $F_0'$ from each of the $\log^2(\nodesize)$ new sketches and $\eta/2$ is from deleting the new spanning forests as they are built from all remaining sketches.

When we are ready to compute $F_i'$, let $\psi = \argmax_j \{i = 0 \mod 2^j\}$. We delete $\bigcup_{j = i - 2^\psi}^{i-1} F_j'$ from $\bigcup_{i}^{i+2^\psi-1} \sketch_i'(\graph)$. For example, we begin computing $F_2'$ by deleting $F_0' \cup F_1'$ both from $\sketch_2(\graph)$ and from $\sketch_3(\graph)$. Later, when we want to extract $F_3'$ from $\sketch_3(\graph)$, we need only delete $F_2'$ from it because the other deletions have already been done. Figure ~\ref{fig:k-connect} summarizes this deletion schedule.

Now we can analyze the I/O cost for each deletion. For simplicity we will overestimate: Round up $k' = k/\log^2(\nodesize)$ to the next power of 2. Each $F_i'$ pays an $\eta/2$ I/O cost for its local deletions, for a total of $k'\eta/2$. Now consider the block deletions. There is one block deletion of size $(k'\eta/2)$ for $F_{k'/2}$, two deletions of size $(k'\eta/4)$ for $F_{k'/4}$ and $F_{3k'/4}$, etc. The total cost for these deletions is therefore \begin{align*}
    \sum_{j<\log(k')} &k' 2^{j} 2^{-(j+1)} \eta 
    \\&= \frac{k'\eta}{2} \log(k') 
    \\ &= \eta \left( \frac{k}{2\log^2(\nodesize)} \log\left(\frac{k}{\log^2(\nodesize)}\right) \right) \\
    &= \frac{k}{2}(\log(k) - 2\log\log(\nodesize))\scan(\nodesize\log^2(\nodesize)). 
\end{align*} 

Therefore the total I/O cost to produce $H$ is $O(k\sort(\nodesize\log^2(\nodesize)) + k\log(k)\scan(\nodesize\log^2(\nodesize)))$.

\begin{figure}
    \centering
    \includegraphics[width=.50\textwidth]{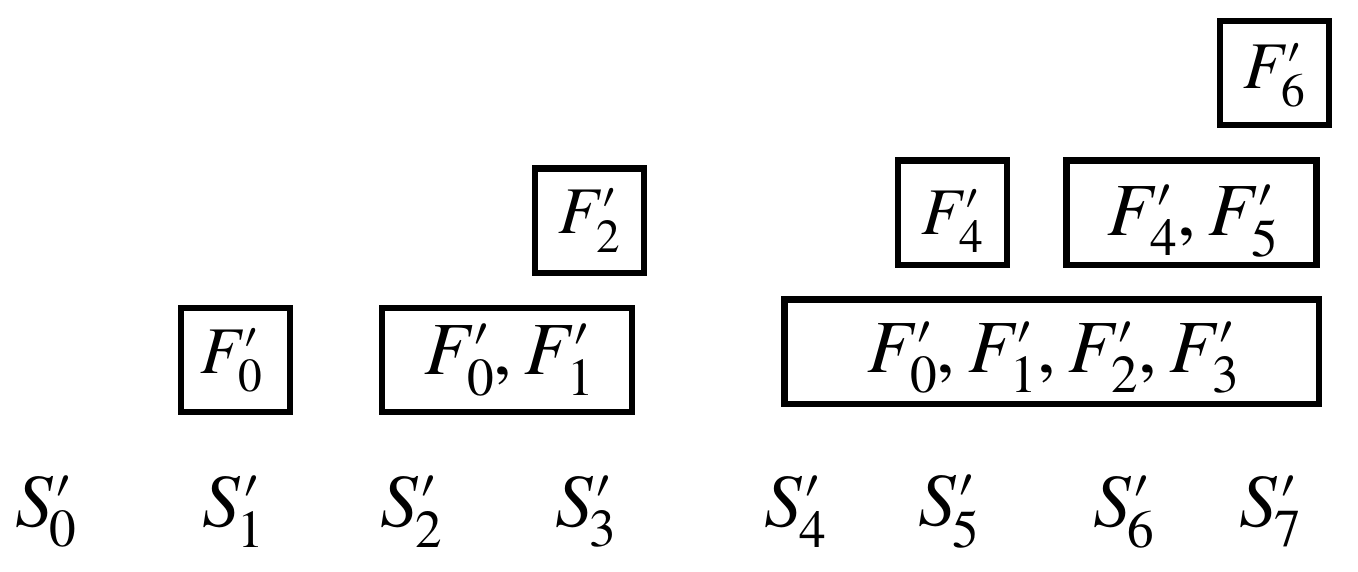}
    \caption{Deletion schedule for k-connectivity query procedure when k = 8. Each rectangle represents a block deletion operation over all the sketches it covers and the labels on the rectangle indicate the spanning forests which are deleted during that block deletion. For example, before querying $\sketch_4'(\graph)$ the four spanning forests $\bigcup_{i=0}^3 F_i$ are deleted from the four sketches $\bigcup_{i=4}^7 \sketch_i$, and before querying $\sketch_5'(\graph)$, only $F_4'$ is deleted.}
    \label{fig:k-connect}
\end{figure}
 
Finally, Geissman and Gianinazzi provide a cache-oblivious algorithm for exact min cut which uses $O(\sort(\edgesize)\log^4(\nodesize))$ I/Os ~\cite{extmincut}. We use this algorithm to determine the edge connectivity of $H$. Since $H$ has at most $k(\nodesize-1)$ edges, the I/O complexity of this step is $O(\sort(k\nodesize)\log^4(\nodesize))$.
\end{proof}

\paragraph{Approximating the minimum cut.}
Ahn \etal~\cite{Ahn2012_2} provide a $O(\epsilon^{-2}\nodesize\log^4(\nodesize))$-space single-pass streaming algorithm for $(1+\epsilon)$-approximating the minimum cut. We summarize it here.

Define $\graph_0 = \graph$ and form $\graph_i \subset \graph_{i-1}$ by deleting each edge in $\graph_{i-1}$ independently with probability 1/2 for each $i \in [O(\log(\nodesize))]$.
For each $\graph_i$, construct a $k = O(\epsilon^{-2}\log(\nodesize))$-skeleton $H_i$ using Ahn \etal's k-connectivity algorithm~\cite{Ahn2012}. The authors show that $\lambda(\graph) \leq 2^j\lambda(H_j) \leq (1 + \epsilon) \lambda(\graph)$ for $j = \min\{i: \lambda(H_i) < k\}$ where $\lambda(D)$ denotes the minimum cut of graph $D$.
Therefore, returning $2^j\lambda(H_j)$ gives the desired approximation to $\lambda(\graph)$. Note that while the algorithm returns a vertex set $S$ such that the cut $(S, \nodes \setminus S)$ has weight no more than $(1+\epsilon)\lambda(\graph)$, it cannot be used to recover the set of edges across $(S, \nodes \setminus S)$. 

We can obtain a \modelname algorithm by applying Theorem~\ref{thm:ingest} and Theorem~\ref{thm:kconnect} to the above sketch, and then applying the external-memory exact min cut algorithm of Geissman and Gianinazzi~\cite{extmincut} to find the minimum cuts of each $(H_i)$. This gives the following theorem.

\begin{theorem}\label{thm:mincut}
There exists a \modelname algorithm for $(1 + \epsilon)$-approximating the minimum cut of $\graph$ w.h.p. that uses  $O(\epsilon^{-2}\nodesize\log^4(\nodesize))$-space and 
\begin{align*} O(&\vsketch(\streamlength, \nodesize, \epsilon^{-2}\log^4 \nodesize) 
\\& + \loglog\nodesize \cdot \ksketch(\nodesize, \epsilon^{-2}\log(\nodesize))
\\& + \log^4(\nodesize)\loglog(\nodesize)\sort(\epsilon^{-2}\nodesize\log(\nodesize)))
\end{align*}
I/Os.
\end{theorem}

\begin{proof}
By Theorem~\ref{thm:kconnect} with $k = \epsilon^{-2}\log\nodesize$, constructing each $\kconsketch(\graph_i)$ takes $O(\epsilon^{-2}\nodesize \log^3(\nodesize))$ space, and there are $O(\log\nodesize)$ subgraphs $\graph_i$.

By Theorem ~\ref{thm:ingest} it costs $\vsketch(\streamlength, \nodesize, \epsilon^{-2}\log^4 \nodesize)$ I/Os to ingest the stream and compute $\kconsketch(\graph_i)$, since the total sketch size per vertex is $\phi = \log^2\nodesize \cdot \epsilon^{-2}\log\nodesize\cdot\log\nodesize$.

Constructing $H_i$ from the sketches for some $i$ takes 
$O(\ksketch(\nodesize, \epsilon^{-2}\log(\nodesize)))$ 
I/Os, by Theorem~\ref{thm:kconnect}.
To find $\lambda(H_i)$ we apply Geissman and Gianinazzi`s exact min cut algorithm which uses $O(\sort(\edgesize)\log^4(\nodesize))$ I/Os ~\cite{extmincut}. 
Since each $H_i$ has a maximum of $O(k\nodesize) = O(\epsilon^{-2}\nodesize\log(\nodesize))$ edges, we can compute $\lambda(H_i)$ in $O(\log^4(\nodesize)\sort(\epsilon^{-2}\nodesize\log(\nodesize)))$ I/Os. 

We perform a binary search through the certificates $H_i$ to find $j = \min\{i: \lambda(H_i) < k\}$. Thus, the total I/O cost to compute $\lambda(\graph)$ from the minimum cut sketch is 
$O(\loglog\nodesize \cdot\ksketch(\nodesize,\epsilon^{-2}\log(\nodesize)) + \log^4(\nodesize))\loglog(\nodesize)\sort(\epsilon^{-2}\nodesize\log(\nodesize)))$.
\end{proof}

\paragraph{Returning the Edges Crossing a Minimum Cut}
In the external-memory model, where we retain access to all the edges in $\graph$, we can recover all of the edges in the approximate minimum cut returned by the above algorithm. Let $(S, \nodes \setminus S)$ denote the  $(1+\epsilon)$ minimum cut returned.  Sort the nodes in $S$ in increasing node ID order. Similarly, sort the list of edges $\edges$ in the input graph $\graph$ in increasing node ID order of the left endpoint, that is, each edge $e = (u,v)$ is sorted in increasing node ID order of $u$. Scan through the list of nodes in S and the list of edges simultaneously, marking each edge $e$ s.t. $u\in S$. Next, sort the edge list in increasing node ID order of right endpoint, and mark each edge $e$ s.t. $v \in S$ similarly. Finally, scan through the edge list and return each edge such that exactly one of its endpoints is in $S$. This gives the following corollary to Theorem~\ref{thm:mincut}:

\begin{corollary}
There exists a 
\begin{align*} O(\sort&(\edgesize) + \loglog\nodesize\cdot\ksketch(\nodesize, \epsilon^{-2}\log(\nodesize))
 \\& + \log^4(\nodesize)\loglog(\nodesize)\sort(\epsilon^{-2}\nodesize\log(\nodesize)))\end{align*}
-I/O external-memory algorithm that returns the edges of a cut $(S, \nodes\setminus S)$ that is at most $(1+\epsilon)$ times the weight of the minimum cut w.h.p.
\end{corollary}

\paragraph{Cut sparsifiers.}
We turn our attention to approximating \emph{any} cut value in the graph. Specifically, the task is to find a \defn{$\varepsilon-$cut sparsifier} $\cert$ of $\graph$, that is, a weighted subgraph $\cert = (\nodes, \edges', w)$ is an $\varepsilon$-cut sparsifier for $\graph$ if $\forall S \subset \nodes$, $(1-\epsilon)\lambda_{S}(\cert) \leq \lambda_{S}(\graph) \leq (1+\epsilon)\lambda_{S}(\cert)$.

Ahn\etal~\cite{Ahn2012_2} provide a semi-streaming algorithm for constructing a cut sparsifier. As in the algorithm for approximating the minimum cut, define $\graph_0 = \graph$, and then graphs $\graph_i \subset \graph_{i-1}$ are formed by deleting each edge in $\graph_{i-1}$ independently with probability 1/2, for each $i \in [O(\log(\nodesize))]$. For each such $i$, construct $\cert_i$, a $k = O(\epsilon^{-2}\log^2(\nodesize))$-connectivity certificate of $\graph_i$. Then a post-processing step decides for each edge $e \in \bigcup_i \cert_i$ whether to add $e$ to the sparsifier $\cert$, as follows. For each $e$, compute $j(e)=\min\{i: \lambda_e(\cert_i) < k\}$
. Then $e$ is added to $\cert$ with weight $2^{j(e)}$ if and only if $e \in \cert_j$.  $\cert$ is returned as the desired cut sparsifier.

We need an I/O efficient way to compute $\lambda_e$ for all the edges in each $\cert_i$. We make use of Laxhuber \etal's $\epsilon$-approximate max flow algorithm~\cite{2021-CE-Maxflow}, which has I/O cost that is proportional to $\edgesize^{1+o(1)}$. By using sketching to sparsify the graph while preserving scaled cut values, we reduce both the number of max flow computations as well as their individual cost by reducing the number of edges by a $\softO(\nodesize)$ factor.

\begin{theorem}\label{thm:simple}
There exists a $O(\epsilon^{-2}\nodesize\log^5(\nodesize))$-space and $O(\vsketch(\streamlength,\nodesize,\varepsilon^{-2}\log^5(\nodesize)) + \log\log(\nodesize)(\epsilon^{-10}\nodesize^2\log^{11}(\nodesize))^{1+o(1)}/B)$-I/O algorithm for constructing a $(1+\epsilon)$-cut sparsifier of graph $\graph$ w.h.p.
\end{theorem}

\begin{proof}
By the bounds in Theorem~\ref{thm:kconnect}, we can ingest the stream and construct the $\cert_i$'s using $O(\vsketch(\streamlength,\nodesize,\varepsilon^{-2}\log^5(\nodesize)) + \log(\nodesize)\ksketch(\nodesize,\varepsilon^{-2}\log^2(\nodesize)))$ I/Os and $O(\epsilon^{-2}\nodesize\log^5(\nodesize))$ space.

For each edge $e \in \bigcup_i \cert_i$, to compute $j(e)$ we use an external-memory $(1-\varepsilon')$-approximation algorithm for maximum flow on undirected graphs with polynomially bounded edge weights due to Laxhuber~\cite{2021-CE-Maxflow}. On a graph with $m$ edges, this algorithm runs in $O(\scan(m^{1+1/\psi})\varepsilon'^{-3})$ I/Os for any constant $\psi \geq 1$. By setting $\varepsilon' = 1/2k = O(\varepsilon^2/\log^2(\nodesize))$, we ensure that this maximum flow algorithm returns the exact value of $\lambda_e(\cert_i)$ provided that $\lambda_e(\cert_i) \leq k$. And if $\lambda_e(\cert_i) \geq k$, the maximum flow returned will also be $\geq k$. By performing a binary search over the $\cert_i$s we can compute $j(e)$ for some $e$ in 
$O(\log\log(\nodesize)\epsilon^{-6}\log^6(\nodesize)(\varepsilon^{-2}\nodesize\log^2(\nodesize))^{1+o(1)}/B)$ I/Os. Since there are $O(\varepsilon^{-2}\nodesize\log^3(\nodesize))$ edges in $\bigcup_{i}\cert_i$, the total I/O cost to compute $j(e)$ for all $e$ is $O(\log\log(\nodesize)\epsilon^{-8}\log^9(\nodesize)(\varepsilon^{-2}\nodesize^2\log^2(\nodesize))^{1+o(1)}/B) = O(\nodesize^{2+o(1)}/\blocksize\cdot \poly(\log\nodesize,\varepsilon^{-1}))$ I/Os.

Finally, once we have computed $j(e)$ for each $e$ in $\bigcup_{i}\cert_i$, we can sort the edges by their $j$ values and determine whether $e \in \cert_{j(e)}$ for all $e$ in $O(\sort(\epsilon^{-2}\nodesize\log^2(\nodesize))$ I/Os.
\end{proof}

Once we have the cut sparsifier, we can use it to approximately answer s-t min cut queries with Laxhuber's max flow algorithm~\cite{2021-CE-Maxflow}. If we want a $(1 + \varepsilon)$ approximation overall, we must set $\varepsilon' = \sqrt{1+ \varepsilon} -1 = O(\varepsilon^{1/2})$ for both the cut sparsifier algorithm and the max flow algorithm. This yields the following corollary:

\begin{corollary}
There exists an algorithm to find $x$ different $s-t$ min cuts on a graph $\graph$ w.h.p. using 
\begin{align*} O(&\vsketch(\streamlength,\nodesize,\varepsilon^{-4}\log^5(\nodesize)) \\& + \log\log(\nodesize)\epsilon^{-20}\log^{11}\nodesize)\scan(\nodesize^2) \\& + x\varepsilon^{-6}\scan(\varepsilon^{-4}\nodesize\log^3(\nodesize))\end{align*} 
I/Os in the external-memory model.
\end{corollary}

\paragraph{Densest subgraph.}
McGregor \etal \cite{McGregorTVV15} give an algorithm for $(1+\epsilon)$-approximating the density $d^*(\graph)$ of the densest subgraph of graph $\graph$. The main idea is to create a subgraph $\cert$, which subsamples each edge in $\graph$ independently with probability $p = \frac{\nodesize\log\nodesize}{\epsilon^2 \edgesize}$, despite the fact that true value of $\edgesize$ (and therefore $p$) is not known until the end of the stream. They show that $\frac{1}{p}d^*(\cert)$ approximates $d^*(\graph)$ to within a factor of $(1+\epsilon)$ with high probability. The density of the densest subgraph of $\cert$ is computed by a black-box algorithm.

The following is performed $O(\log\nodesize)$ times independently in parallel. Before the stream, partition the potential edges of the graph into $\Theta(\epsilon^{-2}\nodesize)$ buckets using pairwise independent hash functions. Insert arriving edges into the $O(\log\nodesize)$ $\ell_0$-sketches corresponding to their bucket. 

In the post-processing step, compute $p$ based on the final number of edges $\edgesize$ that are present in the graph. Then simulate sampling each edge independently with probability $p$ as follows.
For the $i$th bucket in the $j$th partition, randomly draw $X_{ij} \sim Bin(E_{ij},p)$, where $E_{ij}$ is the number of edges present in the bucket at the end of the stream. 
Then, select $X_{ij}$ edges uniformly without replacement by querying each of the first $X_{ij}$ of the bucket's sketches in sequence to produce one edge each, where any queried edges are deleted from all subsequent sketches before querying the next sketch. 
This is performed only for buckets that are `small', i.e., those that contain at most $4\epsilon^2 \edgesize/\nodesize$ edges. The parallel partitions are performed to ensure that every edge is in some small bucket with high probability. Finally, the union of the queried edges from makes up $\cert$.

The above algorithm is not a vertex-based algorithm. However, we show below that it is possible to partition the edges once using a $\Theta(\log\nodesize)$-wise independent hash function such that every bucket is small with high probability. Now we maintain $O(\log^2\nodesize)$ $\ell_0$-sketches for each bucket. This ensures that edges only have to be added to the sketches for their one corresponding bucket, which can be stored contiguously. This allows us to apply Theorem ~\ref{thm:ingest}.

For computing the densest subgraph of $\cert$, we use the algorithm of Charikar~\cite{densest_greedy}, as detailed in the proof of Theorem~\ref{thm:densest}.  This provides a 2-approximation to $d^*(\cert)$, resulting in a $2(1+\epsilon)$-approximation overall.

\begin{theorem}\label{thm:densest}
For a graph $\graph = (\nodes,\edges)$ and $\epsilon > 0$ satisfying $\epsilon^2\edgesize/\nodesize = \Omega(\log\nodesize)$, there exists a $O(\epsilon^{-2}\nodesize\log^3\nodesize)$-space and $O(\vsketch(\streamlength,\epsilon^{-2}\nodesize,\log^3\nodesize) + \nodesize\sort(\epsilon^{-2}\nodesize\log^2\nodesize))$-I/O
  algorithm for $2(1+\epsilon)$ approximating the density of the densest subgraph of a graph $\graph$ w.h.p.
\end{theorem}

\begin{proof}
By Proposition 2.10 in \cite{assadi2023tight}, our modified algorithm produces small buckets w.h.p, since the expected number of edges in each bucket is $\Omega(\log\nodesize)$. Therefore, this algorithm correctly solves the densest subgraph problem.

Our algorithm uses $O(\log^2\nodesize)$ sketches for each of the $\Theta(\epsilon^{-2}\nodesize)$ buckets, for $O(\epsilon^{-2}\nodesize\log^3\nodesize)$ total space. Each bucket must also maintain a counter for the number of edges that have been inserted (and not later deleted). This can be done with an additional $O(\log\nodesize)$-space per bucket. 

Since each edge update only affects the sketches for its bucket, we can essentially treat this as a vertex-based sketching algorithm, where each of the $\Theta(\epsilon^{-2}\nodesize)$ buckets functions as a vertex. Therefore, we can ingest the stream in $\vsketch(\streamlength,\epsilon^{-2}\nodesize,\log^3\nodesize)$ by Theorem~\ref{thm:ingest}. 

In post processing, we must recover an edge from sketch one at a time, and delete them from subsequent sketches for the bucket. The $O(\log^2\nodesize)$ sketches for a bucket fit in memory, so the deletion incurs no additional I/Os. Thus we can recover all sampled edges by scanning through the sketches once in $O(\scan(\epsilon^{-2}\nodesize\log^3\nodesize))$ I/Os, which is dominated by the ingestion cost. 

To compute the densest subgraph in $\cert$, we use Charikar's greedy peeling algorithm ~\cite{densest_greedy}, which gives a 2-approximation to the densest subgraph. The algorithm iteratively removes the lowest degree vertex from the graph, as well as all incident edges, to produce a set of induced subgraphs down to a singleton vertex. The algorithm returns the densest of these subgraphs. 

We represent $\cert$ as an adjacency list on disk: The vertices are listed in increasing degree order. Each vertex is followed by a list of its neighbors. This can be formed in $\sort(\epsilon^{-2}\nodesize\log^2\nodesize)$ I/Os to sort the edges of $\cert$. Remove the lowest degree node from the list, and delete it from its neighbors in $\scan(\epsilon^{-2}\nodesize\log^2\nodesize)$ I/Os. Sort the list so that it is still in increasing degree order in $\sort(\epsilon^{-2}\nodesize\log^2\nodesize)$ I/Os. Since this is repeated $\nodesize$ times, the total I/O cost is $\nodesize\sort(\epsilon^{-2}\nodesize\log^2\nodesize)$.
\end{proof}

 \section{Related Work}

\paragraph{Graph semi-streaming} The graph semi-streaming literature includes both insert only streaming algorithms and fully dynamic (insert and delete) algorithms. 
The results we enumerate here are restricted to algorithms that require only a single pass over the edge stream.
Fully dynamic semi-streaming algorithms include connectivity, bipartiteness testing, cut sparsification, spectral sparsification, approximate minimum spanning tree~\cite{mcgregor2014graph}, approximate hierarchical clustering~\cite{agarwal22sublinear}, vertex coloring~\cite{assadi2019sublinear}, vertex cover~\cite{chitnis2014parameterized}, approximating the minimum cut, and approximating the number of sub-graphs isomorphic to a graph $H$~\cite{Ahn2012_2}.
Semi-streaming algorithms for insert only streams include matching, diameter, spanners~\cite{semistreaming1}, exact minimum spanning tree~\cite{mcgregor2014graph}, and link prediction~\cite{li2018streaming}.

\paragraph{External-memory graph algorithms} 
Chiang \etal~\cite{ChiangGoGr95} presented the first external-memory algorithms for various classical graph problems, including  connected components, minimum spanning forest, list ranking, and Euler tour trees. 
Kumar and Schwabe~\cite{kumarSc96} presented improved bounds for computing minimum spanning forests and gave an algorithm for single-source shortest paths in external memory.
Arge \etal~\cite{ArgeBT04} further improved the bounds for minimum spanning forests, and Meyer and Zeh~\cite{MeyerZ12} showed improved bounds for single-source shortest paths.
Arge~\cite{Arge95thebuffer} introduced the buffer tree, which is useful in a variety of graph problems.
Munagala and Ranade~\cite{munagalaRa99} provided an external-memory breadth-first search algorithm which was later improved by Mehlhorn and Meyer~\cite{mehlhornMe02}. 
Meyer~\cite{Meyer08} gave an algorithm for approximating the diameter of sparse graphs.

 \section{Conclusion}
In this paper we introduce the \modelname model, which combines the stream input and limited space of the semi-streaming model with the block-access constraint of the external-memory model. 

We present a general transformation from any vertex-based sketch algorithm in the semi-streaming model to one which a low sketching cost in the \modelname model. We complement this transformation with a I/O lower bound for  sketching the input stream. For some algorithms, these bounds are tight; for others there is a $O(\log\nodesize)$ gap.

We present several techniques for minimizing the extraction cost. We show how to I/O-efficiently extract many mutually edge-disjoint spanning forests from a k-connectivity, min cut, or cut sparsifier sketch. We also present new external-memory graph algorithms for densest subgraph and cut sparsification. These algorithms have low I/O complexity on sparse graphs and high I/O complexity on dense graphs, but because we sparsify the input graph via sketching, our result is an algorithm that has low I/O complexity on any graph regardless of density.

Putting these techniques together, we present \modelname algorithms for connectivity, hypergraph connectivity, minimum cut, cut sparsification, bipartiteness testing, minimum spanning tree, and densest subgraph. For many of these problems, our \modelname algorithms outperform the state of the art in sketching and external-memory graph algorithms.

The field of semi-streaming has had the problem that the algorithms developed in the model are generally too big for RAM and too random for SSD. This barrier prevents most graph sketch algorithms from running on today's hardware, making them useless for any reasonable application. External semi-streaming algorithms get around this barrier because they are small enough to store on SSD and they make mostly sequential accesses. This means that instead of having to wait decades for these algorithms to be useful, we may be able to make use of graph sketching on today's hardware.

Given the results in this paper, we believe that I/O complexity should be treated as a first-class citizen in the design and analysis of semi-streaming algorithms. The transformations in this paper for vertex-based sketches makes it more feasible, and in some cases even trivial, to design \modelname algorithms. We also believe that sketching is a powerful technique for designing external-memory graph algorithms, even outside of a streaming setting.

Finally we note that the fields of external memory and semi-streaming have been parallel but separate ways of dealing with massive data. In this paper we show that each field can contribute to the other in important ways. \bibliographystyle{abbrv}

\newpage
\begin{appendix}\label{appendix}

\section{Case Study: k-Connectivity} \label{app:case}
In this section, we expand on why graph sketching is infeasible using today's hardware with a case study of k-connectivity.

Consider the following back-of-the-envelope calculation.  Ahn \etal's~\cite{Ahn2012_2} sketch for testing k-connectivity has size $O(k\nodesize \log^3\nodesize)$ bits. Assuming 128GB of RAM and $k = \log\nodesize$, the largest graph whose sketch fits in RAM has $114,000$ vertices.\footnote{To compute the actual size of the sketch, including constants, we extrapolate from the sketch size in Tench \etal's implementation~\cite{tenchwevz22gz} of Ahn \etal's~\cite{Ahn2012} connectivity sketch.}

However, for graphs of this size, sketching doesn't even save space.
The adjacency list of even a complete graph on $114,000$ vertices has size at most 13GiB, while the k-connectivity sketch has size 128GiB---an order of magnitude larger. 
Despite the fact that the sketch is asymptotically smaller than an adjacency list, at this input size its high constant factors and the $\log^3\nodesize$ term dominate the space cost.

Thus, the asymptotic advantages of sketching for semi-streaming algorithms~\cite{Ahn2012,AhnGM13,Ahn2012_2,KapralovLMMS13,KapralovW14,assadi2023tight,pagh2012colorful,McGregorVV16} so far have not translated into space savings in implementation. When might we expect empirical advantages? 

These advantages are unlikely to arrive for a long time, unfortunately.  For example, for the k-connectivity sketch to save space, we need graphs with at least $\nodesize = 33,000,000$ vertices and at this size, the k-connected-components sketch has size $118$TiB. Further, the $\nodesize$ and $\log\nodesize$ factors in the k-connectivity space complexity cannot be improved asymptotically due to a lower bound of Nelson and Yu~\cite{nelson2019optimal}. So sketching will not save us space for this problem until we have RAM sizes in the hundreds of TiBs.

\section{Hardware Trends in Solid State Drives}
\label{app:disk}
The \modelname model requires storing most data structures on disk rather than RAM. This approach helps because solid state drives (SSDs) are big enough to store sketches. High-performance enterprise grade SSDs can be bought in sizes of 64 TBs~\cite{solidigm24} and consumer grade SSDs reach up to 8 TB in size~\cite{clark23}. Many such disks can be used on a single computer.  But the approach appears to be a step backwards because the underlying motivation of streaming and semi-streaming is that disks are too slow to support effective computation on graph data~\cite{mcgregor2014graph}.

However, in the decades since the introduction of the streaming model, the cost of accessing disk has changed. It is no longer exactly true that disks are too slow---in fact they are fast enough if accessed in the right way.
Notably, the sequential bandwidth of modern SSDs is very fast. In fact, SSD sequential bandwidth is already nearly as fast as RAM random-access bandwidth~\cite{bandwidth2}---and SSD bandwidth is growing faster~\cite{bandwidth}.

Unfortunately, most existing graph-sketching algorithms use techniques such as hashing to random locations \cite{mcgregor2014graph,Ahn2012,KapralovLMMS13,GuhaMT15,assadi2023tight}, meaning that most accesses that the algorithm makes will be random accesses. Using disks 
the way most graph-sketching algorithms need
is too slow because SSD \emph{random-access bandwidth} is much slower than that of RAM~\cite{bandwidth}. 
The result is that an off-the-shelf sketching algorithm designed for RAM will run orders-of-magnitude more slowly on an SSD, making it useless for any reasonable application.

\section{Summary of Existing Graph Sketch and External-Memory Algorithms}\label{app:standard}
We summarize the I/O cost of existing graph sketching and external-memory algorithms for the problems we study in this paper in Table~\ref{tab:standard}.
\begin{table*}[t]
\small
  \begin{tabular}{l|c|c}
    \toprule
    \textbf{Algorithm} & \textbf{Standard Sketching I/O Cost} & \textbf{External Memory} \\
    \midrule
Connected Components~\cite{Ahn2012,tenchwevz22gz,ChiangGoGr95} & $O(\sort(\streamlength) + \nodesize \alpha(\nodesize) + (\streamlength + \nodesize) \scan(\log^2 \nodesize))$ & $\sort(\streamlength)$ \\
    
    Bipartiteness Testing~\cite{Ahn2012,ChiangGoGr95} & $O(\streamlength + \nodesize \alpha(\nodesize) + (\streamlength + \nodesize) \scan(\log^2 \nodesize))$ & $\sort(\streamlength)$  \\
    
    Hypergraph Connectivity~\cite{GuhaMT15} &  $O(\streamlength + \nodesize \alpha(\nodesize) + (\streamlength + \nodesize) \scan(r\log^2 \nodesize))$ & no nontrivial algorithm \\
    
    $\epsilon$-Approx. MST Weight~\cite{Ahn2012,ChiangGoGr95} & $O(\streamlength + \varepsilon^{-1} \nodesize \alpha(\nodesize)   + (\streamlength + \nodesize) \scan(\varepsilon^{-1} \log^2 \nodesize))$ & $\sort(\streamlength)$ (exact) \\
    
    $k$-Connectivity~\cite{Ahn2012_2,extmincut}  & $O(\streamlength + k\nodesize \alpha(\nodesize) + (k\streamlength + k^2\nodesize) \scan(\log^2 \nodesize))$ & $O(\sort(\streamlength) \log^4\nodesize)$ (min-cut)  \\
    
    $\epsilon$-Approx. Minimum Cut~\cite{Ahn2012_2,extmincut} & 
    $O(\streamlength + \nodesize \varepsilon^{-2} \log^2 \nodesize \alpha(\nodesize) + (\streamlength \varepsilon^{-2}$ & $O(\sort(\streamlength) \log^4\nodesize$) (exact) \\

    ~ & $\log^2 \nodesize + \nodesize \varepsilon^{-4}\log^3 \nodesize) \scan(\log^2 \nodesize))$ & ~  \\

    $\epsilon$-Cut Sparsifier~\cite{Ahn2012_2} & $O(\streamlength + \nodesize \varepsilon^{-2} \log^3 \nodesize \alpha(\nodesize) + (\streamlength \varepsilon^{-2}$ & no nontrivial algorithm \\
    
    ~ & $\log^3 \nodesize + \nodesize \varepsilon^{-4}\log^4 \nodesize) \scan(\log^2 \nodesize))$ & ~  \\
    
    $\epsilon$-Approx. Densest Subgraph~\cite{McGregorTVV15} & $O(\streamlength\log\nodesize + (\streamlength + V) \scan(\varepsilon^{-2} \log^2\nodesize))$ & no nontrivial algorithm \\
  \bottomrule
\end{tabular}
\caption{Bounds for naively using known sketching algorithms in external memory and for the external memory algorithms that solve these problems. The input to all algorithms is a dynamic graph stream of length $\streamlength$ on $\nodesize$ vertices. Note that the sketching column indicates the I/Os required to sketch the input stream and to extract a sparser graph from the sketch that approximates the original property in question. We do not include the I/O cost of computing said property on the sparser graph. Also note that all existing EM algorithms for these problems have a space cost of $\Omega(N)$ because they need to store the entire graph. $\alpha(x)$ is the inverse Ackermann function.
}
\label{tab:standard}
\end{table*}

\section{I/O Complexity of \algname}
\label{app:iobad}
In the streaming connectivity problem, stream updates are \emph{fine-grained}: each update represents the insertion or deletion of a single edge.  Since streams are ordered arbitrarily, even a short sequence of stream updates can be highly non-local, inducing  changes throughout the graph.  As a result, \algname and similar graph streaming algorithms do not have good data locality in the worst case. This lack of locality can incur many I/Os and therefore reduce the ingestion rate if the sketches are stored on disk. The I/O cost can be high since ingesting each stream update $(u,v,\Delta)$ requires modifying two poly-logarithmic-sized vertex sketches, and can thus induce a poly-logarithmic number of I/Os.

\begin{observation}
In the hybrid graph streaming setting with $\memsize = o(\nodesize \log^2(\nodesize))$ RAM and $D = \Omega(\nodesize \log^2(\nodesize))$ disk, \algname uses $\Omega(1)$ I/Os per edge update, and processing the entire stream of length $\streamlength$ uses $\Omega(\streamlength) = \Omega(\edgesize)$ I/Os.
\end{observation}

Any sketching algorithm that scales out of core suffers severe performance degradation unless it amortizes the per-update overhead of accessing disk. Such an amortization is not straightforward, since sketching inherently makes use of hashing and as a result induces many random accesses, which are slow on persistent storage.
Tench et al.~\cite{tenchwevz22gz} give an external-memory sketching algorithm for connected components which amortizes disk access costs, even for adversarial graph streams. The result is an algorithm that is both space-efficient and I/O-efficient. 
We restate below their results on the space- and I/O- complexity of first reading in the stream, and then computing the connected components of the resulting graph.

\begin{lemma}[Restated from~\cite{tenchwevz22gz} Lemma 4]
    The algorithm from~\cite{tenchwevz22gz} uses $O(\nodesize \log^2 \nodesize)$ space and $O(\sort(\streamlength) + \scan(\nodesize\log^2\nodesize))$ I/Os to process the stream updates.
\end{lemma}
 
\begin{lemma}[Restated from~\cite{tenchwevz22gz} Lemma 5]
    Once all stream updates have been processed, the algorithm from~\cite{tenchwevz22gz} computes connected components using $O((\nodesize/B) \log^2(\nodesize) + \nodesize \alpha(\nodesize))$ I/Os.
\end{lemma}

\end{appendix}

 \end{document}